\documentclass[journal,12pt]{IEEEtran}

\ifCLASSINFOpdf
  \usepackage[pdftex]{graphicx}
  \DeclareGraphicsExtensions{.pdf}
\else
\fi
\usepackage[cmex10]{amsmath}
\usepackage{setspace}
\usepackage{multicol}
\usepackage{caption}
\usepackage{amsthm}
\newtheorem{mydef}{Definition}
\newtheorem{lem}{Lemma} 
\newtheorem{theo}{Theorem} 

\begin{document}
\linespread{1}
%

\title{The Chief Security Officer Problem}

\author{\IEEEauthorblockN{Kamesh Namuduri$^{1}$, Li Li$^{1}$, Mahadevan Gomathisankaran$^{2}$, Murali Varanasi $^{1}$\\}
\IEEEauthorblockA{$^{1}$Dept. of Electrical Engineering\\
$^{2}$Dept. of Computer Science and Engineering\\
University of North Texas\\
kamesh.namuduri@unt.edu, lili@my.unt.edu, mgomathi@unt.edu, murali.varanasi@unt.edu}
}

\maketitle
\IEEEpeerreviewmaketitle

\begin{abstract}

This paper presents the chief security officer (CSO) problem, defines its scope, and investigates several important research questions related within the scope. The CSO problem is defined based on the concept of secrecy capacity of wireless communication channels. It is also related to the chief Estimation/Executive Officer (CEO) problem that has been well studied in information theory.
The CSO problem consists of a CSO, several
agents capable of having two-way communication with the CSO, and a group of eavesdroppers. There are two scenarios in the CSO problem; one in which agents are not allowed to cooperate with one another and the other in which agents are allowed to cooperate with one another. While there are several research questions relevant to the CSO problem,
this paper focusses on the following and provides answers:  (1) How much information can be exchanged back and forth between the CSO and the agents without leaking any information to the eavesdroppers? (2) What is the power allocation strategy that the CSO needs to follow so as to maximize the secrecy capacity?  (3) How can agents cooperate with one another in order to increase the overall secrecy capacity?. 
\end{abstract}


\section{Introduction}

The Chief Security Officer (CSO) problem consists of a CSO, several
agents capable of having two-way communication with the CSO, and a group of eavesdroppers \footnote{The problem was originally presented in a 45th Annual Conference on Information Sciences and Systems (CISS), 2011.}\cite{namuduri}. The objective is to estimate the secrecy capacity of the CSO channel which refers to the maximum achievable rate of information transfer between the CSO and the agents without leaking any of the information to the eavesdropper.

The CSO problem is mainly based on the secrecy capacity of wireless communication channels that has been investigated over the past several years
\cite{wyner1,korner1,cheong1,bergmans1,csaiszar1,barros1,liang1}.
In a standard Alice-Bob-Eve model, secrecy capacity is defined as the maximum amount of
information that Alice can convey to Bob while keeping it
secret from Eve, the eavesdropper. Secrecy capacities of different
types of wireless channels including noisy, fading, and multiple input and
multiple output (MIMO) channels have been investigated over the past few years
\cite{gopala1,negi1,Shafiee1,zhangli,oggier1,thangaraj1}.

The CSO problem is also related to the Chief Estimation/Executive Officer (CEO) problem \cite{Berger1} which seeks to find tradeoffs between the sum rate of information transfer and the total distortion
between the true and estimated values of the source. In the CEO problem, the source is represented as a random vector, the observations are assumed to be noisy, and the communication
channels between the agents and the CEO are assumed to be noise-free. The CEO receives
observations from a number of agents. An important aspect of the CEO problem is that the \emph{agents are not allowed to communicate among
themselves}. The CEO problem has been
used to model information and decision fusion concepts
in sensor networks (see for example, \cite{Razi2011}). Characterization of rate-distortion region in distributed source coding has been investigated through CEO problem and its variants such as the multi-terminal source coding \cite{Oohama1997,Oohama2005}. Specific cases within multi-terminal source coding, for example, the quadratic Gaussian multi-terminal source coding \cite{Wagner2008}) have also been investigated in the literature. 

The CEO  problem is related to source coding, whereas, the CSO problem is related to channel coding. The CEO problem is discussed here only to bring out duality characteristics that typically exist between source coding and channel coding. Given that the CEO problem has been investigated in the literature over two decades, its analysis helps us gain insights into the CSO problem. Similarly, the new results that are discussed in our work may provide new insights into the CEO problem. Detailed analysis of these duality characteristics is beyond the scope of this paper and will be reported in our future work.

\subsection{The CSO Problem}

The CSO problem consists of a CSO, a group of \textit{N} agents working
for the CSO gathering information, and another group of \textit{M} eavesdroppers capable of listening to all conversations as shown in Fig. \ref{fig1}. One can consider two scenarios: In the first scenario, the
agents are \emph{not} \emph{allowed} to communicate with one another. In the second scenario, the agents are \emph{allowed to communicate} with one another.

The CSO problem includes two types of channels: the downlink broadcast channel (BC), from CSO to agents, as shown in Fig. \ref{ftwo}, and the uplink multiple access channel (MAC), from agents to CSO, as shown in Fig. \ref{fthree}. In both channels, the messages are denoted by $W=W_{[1,N]}$, the input codewords denoted by $X=X_{[1,N]}$  and the output of the channel is denoted by $Y=Y_{[1,N]}$. The information collected by the eavesdroppers is denoted by $Z=Z_{[1,N]}$. Subscripts [$1 \ldots, N$] are used to represent the users and superscripts [$1, \ldots, n$] are used to represent the number of symbols transmitted or  number of times the channel is used. This common notation is followed for both channels to bring out the duality between them. 

We consider security under information theory paradigm and define secrecy capacity of a channel in terms of the channel transition probability,  probability of error at the legitimate receiver, and the equivocation rate for the eavesdropper \cite{cover2,liang1}. Extending the concept of perfect secrecy as discussed in \cite{liang1}, we consider perfect secrecy when the eavesdropper does not obtain any information about the messages transmitted on the legitimate channel. This happens when the equivocation rate for the eavesdropper is  as much as the maximum rate of transmission on the legitimate channel.

\subsubsection{Downlink Broadcast Channel}
Following the notation used  in \cite{cover2}, the CSO downlink BC channel consists of an input alphabet $\mathcal{X}$, $N$ output alphabets $\mathcal{Y}_1, \mathcal{Y}_2, \ldots, \mathcal{Y}_N$ and a probability transition function defined in Eq. (\ref{eq1}). The downlink BC channel can be viewed as a set of parallel independent Alice-Bob-Eve secrecy channels, one of which is shown in  Fig. \ref{fone}.  In each of these parallel channels, Alice, the CSO, sends $X^{n}$, a sequence of  $n$ symbols, to Bob, the agent. Bob receives a noisy version $Y^{n}$ of $X^{n}$,  while Eve, the eavesdropper, receives another noisy version $Z^{n}$.

The CSO transmits $X^{n}$ to multiple agents on the BC channel. Fig. \ref{ftwo} illustrates the CSO downlink BC channel with two agents. Each agent receives a noisy version of $X^{n}$, which is $Y^{n}_{i}$, $i = 1,2$. The agents transmit their information, $X^{n}_{i}$, $i = 1,2$, to the CSO, on the uplink MAC, as shown in Fig. \ref{fthree}.  In both cases, multiple Eves (eavesdroppers) listen on all channels between the CSO and the agents.

%
The channels are assumed to be independent of each other such that the transition probability distribution function (PDF) of the CSO downlink BC channel is:
\begin{equation}
p(y_{[1,N]}, z_{[1,N]} | x) =
\prod_{i=1}^{N}p(y_i,z_i|x),
\label{eq1}
\end{equation}
where $x_i \in \mathcal{X}_i$, $y_i \in \mathcal{Y}_i$, $z_i \in
\mathcal{Z}_i$, $i = 1, \ldots, N$, and $\mathcal{X}_i$, $\mathcal{Y}_i$, $\mathcal{Z}_i$ represent the alphabet.

\noindent Similarly, the PDF of CSO uplink MAC channel is:
\begin{equation}
p(y, z_{[1,N]} | x_{[1,N]}) =
\prod_{i=1}^{N}p(y,z_i|x_i),
\label{eq2}
\end{equation}

\noindent Both BC and MAC channels in the CSO problem are assumed to be memoryless, i.e.,
\begin{eqnarray}
BC: & p(y^n_{[1,2,\ldots,N]}, z^n_{[1,2,\ldots,N]} | x^n)  = &  \\ \nonumber &\prod_{j=1}^{n}p(y^{j}_{[1,2,\ldots,N]},z^{j}_{[1,2,\ldots,N]}|x^{j}),& \\
MAC: & p(y^n, z^n_{[1,2,\ldots,N]} | x^n_{[1,\ldots,N]}) = & \\ \nonumber &\prod_{j=1}^{n}p(y^{j},z^{j}_{[1,2,\ldots,N]}|x^{j}_{[1,2,\ldots,N]}),&
\end{eqnarray}

\noindent where $n$ represents the number of times the channel is used and $j \in (1,\ldots,n)$ represents a particular instance of usage.

A (($2^{nR_1},2^{nR_2}, \ldots, 2^{nR_N}$), $n$) code for CSO BC channel with independent information consists of an encoder \textit{f} defined by 
$x^n: 2^{nR_1}\times2^{nR_2}\times\ldots2^{nR_N} \rightarrow \mathcal{X}^n$ and $N$ decoders $\phi_i, i \in (1,2,\ldots,N)$ defined by
$\hat{W_1}: \mathcal{Y}_1^n \rightarrow 2^{nR_1}$,
$\hat{W_2}: \mathcal{Y}_2^n \rightarrow 2^{nR_2}$,
$\ldots$,
$\hat{W_N}: \mathcal{Y}_N^n \rightarrow 2^{nR_N}$.
The probability of error $P_e^{(n)}$ is defined to be the probability of the decoded message not equal to the transmitted message, i.e., $P_e^{(n)} = P( \hat{W_1}(Y_1^n) \neq W_1$ or  $\hat{W_2}(Y_2^n) \neq W_2$ or $\ldots$, $\hat{W_N}(Y_N^n) \neq W_N)$,
where the messages ($W_1$, $W_2$, $\ldots$, $W_N$) are assumed to be distributed over $2^{nR_1}\times2^{nR_2}\times\ldots2^{nR_N}$.

\begin{mydef}
The CSO downlink BC channel is said to achieve secrecy rate $R = \sum_{i=1}^{N} R_i$  if there exists a sequence of (($2^{nR_1},2^{nR_2}, \ldots, 2^{nR_N}$), $n$) codes for the (1, N) encoder-decoders set with probability of error $P_e^{(n)}$ and the equivocation rate for the eavesdropper $\lim_{n \to \infty} \frac{1}{n}H(W_i|Z_{[1,N]}^n) \geq R_e \geq R_i$, for all i $\in$ [1,N].
\end{mydef}

The CSO uplink MAC channel consists of $N$ input alphabets  $\mathcal{X}_1, \mathcal{X}_2, \ldots, \mathcal{X}_N$,   one output alphabet $\mathcal{Y}$, and a probability transition function defined in Eq. (\ref{eq2}). A (($2^{nR_1},2^{nR_2}, \ldots, 2^{nR_N}$), $n$) code for CSO MAC channel with independent information consists of $N$ encoders $f_i, i \in (1,2,\ldots,N)$ defined by
$x_1^n: 2^{nR_1} \rightarrow \mathcal{X}_1^n $,
$x_2^n: 2^{nR_2} \rightarrow \mathcal{X}_2^n $,
$\ldots$,
$x_N^n: 2^{nR_N} \rightarrow \mathcal{X}_N^n$, and a decoder $\phi$ defined by 
$y^n: 2^{nR_1}\times2^{nR_2}\times\ldots2^{nR_N} \rightarrow \mathcal{Y}^n$. 
The probability of error $P_e^{(n)}$ is defined to be the probability of the decoded message not equal to the transmitted message, i.e., $P_e^{(n)} = P( \hat{W_1}(X^n) \neq W_1$ or  $\hat{W_2}(X^n) \neq W_2$ or $\ldots$, $\hat{W_N}(X^n) \neq W_N)$,
where the messages ($W_1$, $W_2$, $\ldots$, $W_N$) are assumed to be distributed over $2^{nR_1}\times2^{nR_2}\times\ldots2^{nR_N}$.

\begin{mydef}
The CSO uplink MAC channel is said to achieve secrecy rate $R = \sum_{i=1}^{N} R_i$  if there exists a sequence of (($2^{nR_1},2^{nR_2}, \ldots, 2^{nR_N}$), $n$) codes for the (N, 1) encoders-decoder set with probability of error $P_e^{(n)}$ and the equivocation rate for the eavesdropper $\lim_{n \to \infty} \frac{1}{n}H(W_i|Z_{[1,N]}^n) \geq R_e \geq R_i$, for all i $\in$ [1,N].
\end{mydef}

Since the scope of this paper is limited to memoryless and independent channels, from here onwards, we remove the superscript $j= \in {1,\ldots,n}$, and use Eqs (\ref{eq1}) and (\ref{eq2}) to describe the downlink and uplink CSO channels respectively. We consider, however, both correlated and uncorrelated messages in the CSO system. The analysis begins with the scenario in which the messages $W_{[1,N]}$ transmitted over the channels are assumed to be uncorrelated and it is extended to include correlated messages.

\subsection{Real World Scenarios}
We discuss two real world scenarios that motivate us design the CSO channel model, identify effective ways to achieve its capacity, and gain insights into the model.

\textit{Battlefield Communications}: Consider a group of soldiers gathering information about a target. Each soldier
conveys the information he or she gathers to the CSO. In
turn, the CSO communicates the actions that need to be
taken by each of the soldiers. Two scenarios arise depending on whether peer to peer communication among the soldiers is allowed or not. If peer to peer communication is not allowed, soldiers communicate only with the CSO. If peer to peer communication is allowed, soldiers form an ad hoc network, use cooperative strategies to communicate with the CSO. Cooperative communication strategies enhance secrecy capacity of the network.

\textit{Drug Dealer's Cocktail Party}: Consider a cocktail party (similar to the one described in \cite{cover2}), that is hosted by a drug dealer. Assume that a group of
agents  join this party disguising as guests, operate covertly, gather information, and pass it to the CSO. As in the previous example, the network structure created by agents and the CSO varies depending on whether agents are allowed to cooperate among themselves or not. If cooperation among the agents is allowed, the secrecy capacity of the network will be higher compared to the scenario in which cooperation is not permitted. 

\subsection{Common Objectives}

The CSO problem abstracts the underlying concept in both scenarios described above. The objective in each scenario is to pass information along to the CSO without eavesdropper's knowledge. The downlink channel from the CSO to the agents is also as important as the uplink channel. The two scenarios showcase the possible variations in covert communications in terms of the participants' awareness of the presence of the team members, and possibilities of collaboration among them. They  demonstrate the need for efficient  resource allocation strategies that maximize the information exchanged. They also demonstrate the constraints on covert communications. For example, observation of mere presence of an agent may undermine the overall objective of the mission, and nullify the value of information gained thus far. Real world applications of these fundamental research problems include secret communications in wireless networks, cellular networks, and sensor networks.

\subsection{Central Issues} There are several interesting research
questions that come up in the CSO framework. We begin the
discussion with a set of simple questions related to covert
communications between the CSO and the agents: 1) How much
information the CSO can secretly convey to any one of the agents? 2)
How much information does any one of the agents can pass along to
the CSO secretly? 3) What is the sum of information that the CSO can
secretly convey to all the agents and 4) What is the sum of
information that all agents can collectively and secretly convey  to the CSO ? 5) What is the optimal power allocation
strategy that achieves the secrecy capacity? 6) How can agents cooperate with one another to increase the secrecy capacity?

\subsection{Main Contributions}

Our main contribution is the introduction of the CSO problem and derivation of expressions for the secrecy capacity of CSO broadcast and MAC channels with uncorrelated and correlated messages. Further, we bring out the duality between the secrecy capacities of CSO broadcast and MAC channels.
We discuss the power allocation strategy for the CSO broadcast channel with independent AWGN and fading channels based on the results presented in \cite{zhangli}. In each scenario, we identify applicable existing results in the literature and relate to the CSO problem. We demonstrate that cooperation increases the secrecy capacity of the CSO BC channel. Using a new metric {\it{secrecy efficiency}}, we demonstrate that as much as 50\% improvement in secrecy efficiency can be achieved through the cooperation among the agents. We provide a proof of optimality for the proposed cooperation strategy.

\subsection{Organization}
In section II, results for the CSO problem without cooperation among agents are presented. In section III, a cooperation strategy is presented and results for the CSO problem with cooperation are presented. Section IV provides summary and conclusions. Proof of optimality for the proposed cooperation strategy is presented in appendix A.

\section{A General Model for the CSO Problem}

In this section, we develop information theoretic results for the secrecy capacity of CSO BC and MAC channels. We begin with the BC channel and  discuss the scenario in which the transmitted messages $W_{[1,\ldots,N]}$ on the  BC channel are independent and summarize our results in Lemmas 1 and 2. Then, the scenario in which the transmitted messages are correlated  is discussed in Lemmas 3 and 4. Next, the CSO MAC with correlated messages is discussed in Lemmas 5 and 6. Lemma 7 briefly discusses the duality between MAC and BC. Detailed discussion on the duality aspects in CSO problem is outside of the scope of this paper. Lemmas 8 and 9 discuss the secrecy capacities of the BC channel with additive white Gaussian noise (AWGN) and Rayleigh fading respectively. Theorems 1 and 2 discuss the power allocation strategies.

\begin{lem}
The secrecy capacity ($C_s^{(i,down)}$) of the downlink channel
between the CSO and $i^{th}$ agent is given by
\begin{math} 
C_s^{(i,down)} = \max \limits_{U_i \rightarrow X_i \rightarrow Y_i Z_i}[I(U_i;Y_i) - I(U_i;Z_i)].
\end{math}
\end{lem}

The auxiliary random variable $U_i$ corresponding to channel $i$ will serve as the cloud center distinguishable by both the legitimate receiver as well as the eavesdropper. We refer the reader to \cite{cover1} for an excellent discussion on the auxiliary variable. We use the symbol $U$ to refer to the set of auxiliary variables ($U_1$, $U_2$, $\ldots$, $U_N$) corresponding to channels (1, 2, $\ldots$, N).

\begin{proof}
This lemma follows from the basic definition of a wiretap
channel \cite{cheong1}. The presence of other channels will not change this result as long as the
inputs to each of the channels are independent of each other.
\end{proof}

\begin{lem}
The sum capacity  ($C_s^{(sum,down)}$) of the CSO downlink broadcast channel 
consisting of $N$ parallel channels is given by

\begin{eqnarray} 
\label{Cssum}
&&C_s^{(sum,down)} = \sum_{i=1}^{N}C_s^{(i,down)} \\ \nonumber  
&=&\sum_{i=1}^{N} \max \limits_{U_i \rightarrow X_i \rightarrow Y_i Z_i} [I(U_i;Y_i) - I(U_i;Z_i)].
\end{eqnarray}
\end{lem}

\begin{proof} This lemma follows from the results of secrecy of parallel
channels \cite{liang1}. Here, as in Lemma 1, it is assumed that each channel $i$ represents a point to point communication link between the CSO and $i^{th}$ agent. \end{proof}

Lemmas 1 and 2 are valid for the scenario in which the transmitted messages $W_{[1,N]}$ on the channels are uncorrelated. When the transmitted messages are correlated, both the CSO and the eavesdroppers benefit from the aggregation of inputs received from all agents. Of course, agents also benefit from aggregation provided that they are allowed to communicate among themselves. However, to begin with, we assume that agents do not communicate among themselves and derive the following results.

\begin{lem}
When $W_i$'s are correlated, the secrecy
capacity of the downlink channel between the CSO and any agent $i$ 
denoted by $C_s^{(i,down)}$ is a function of the number of agents (i.e., channels) and is given by 
\begin{math} 
C_s^{(i,down)} = \max \limits_{U_i \rightarrow X_i \rightarrow Y_i Z_i}[I(U_i;Y_i) - I(U_i;Z_{[1,\ldots,N]})].
\end{math}
\end{lem}

\begin{proof} Each time an agent sends information on a channel, the eavesdropper(s) benefits from
aggregating the information they acquires by tapping all channels. Agents lose this advantage if they are not allowed to communicate among themselves.

\begin{eqnarray}
C_s^{(i,down)} &=& \max  \limits_{U_i \rightarrow X_i \rightarrow Y_i Z_i} [I(U_{i};Y_{i}) \\ \nonumber
&& ~ - ~I(U_i;Z_{i}| Z_{[1,..,i-1,i+1,..,N]})] \\ \nonumber
&=& \max \limits_{U_i \rightarrow X_i \rightarrow Y_i Z_i}  [I(U_{i};Y_{i}) -  I(U_i; Z_{[1,\ldots,N]})] \\ \nonumber
&=& \max \limits_{U_i \rightarrow X_i \rightarrow Y_i Z_i} [I(U_{i};Y_{i}) - I(U_i; Z)]
\end{eqnarray}
\end{proof}

As the number of agents increases, the secrecy capacity of the downlink channel between the CSO and any agent $i$ may diminish due to the advantage  eavesdroppers have over the agents.

\begin{lem}
When $W_i$'s are correlated, the sum secrecy
capacity of the downlink channels between the CSO and the agents,
denoted by $C_s^{(sum,down)}$,  is bounded above by
\begin{eqnarray}
C_s^{(sum,down)} &\leq & \max \limits_{U \rightarrow X \rightarrow Y Z}  [I(U;Y) - \\ \nonumber
&& I(U;Z)].
\end{eqnarray}
\end{lem}

\begin{proof} The eavesdroppers  benefit even
more by aggregating the information leaked from all downlink channels which leads to the following result:

\begin{eqnarray}
&&C_s^{(sum,down)} = \sum_{i=1}^{N}C_s^{(i,down)}  \\ \nonumber
&=& \sum_{i=1}^{N} \max  \limits_{U_i \rightarrow X_i \rightarrow Y_i Z_i} [I(U_{i};Y_{i})-I(U_i; Z_{[1,\ldots,N]})] \\ \nonumber
& \leq&  \max \limits_{U_i \rightarrow X_i \rightarrow Y_i Z_i} [ \sum_{i=1}^{N}I(U_{i};Y_{i})- \sum_{i=1}^{N}I(U_{i}; Z_{[1,\ldots,N]}) ] \\ \nonumber
& \leq &  \max \limits_{U \rightarrow X \rightarrow Y Z} [I(U;Y)-I(U; Z)]
\end{eqnarray}
\end{proof}

We move on to the CSO uplink MAC channel. The following two Lemmas discuss the secrecy capacity corresponding to the CSO uplink MAC channels. 

\begin{lem}

When $W_i$'s are correlated, the secrecy
capacity of the uplink channel between any agent $i$ and the CSO
denoted by $C_s^{(i,up)}$ is a function of the number of agents (i.e., channels) and is given by:

\begin{math} 
C_s^{(i,up)} = \max \limits_{U_i \rightarrow X_i \rightarrow Y_i Z_i}[I(U_i;Y_i) - I(U_i;Z_{[1,\ldots,N]})].
\end{math}
\end{lem}

\begin{proof} Each time an agent sends information on a channel, the CSO makes use of the information that has been received from all other agents as side information.
On the same token, the eavesdropper also benefits from
aggregating the information it acquires by tapping all channels.

\begin{eqnarray}
C_s^{(i,up)} &=& \max  \limits_{U_i \rightarrow X_i \rightarrow Y_i Z_i} [I(U_{i};| Y_{[1,..,i-1,i+1,..,N]}) \\ \nonumber
&& ~ - ~I(U_i;Z_{i}| Z_{[1,..,i-1,i+1,..,N]})] \\ \nonumber
&=& \max \limits_{U_i \rightarrow X_i \rightarrow Y_i Z_i}  [I(U_{i};Y_{[1,\ldots,N]}) -  I(U_i; Z_{[1,\ldots,N]})] \\ \nonumber
&=& \max \limits_{U_i \rightarrow X_i \rightarrow Y_i Z_i} [I(U_{i};Y) - I(U_i; Z)]
\end{eqnarray}
\end{proof}

As the number of agents increases, the secrecy capacity of the uplink between the agent and the CSO may increase.

\begin{lem}
When $W_i$'s are correlated, the sum secrecy
capacity of the uplink channels between the agents and the CSO,
denoted by $C_s^{(sum,up)}$,  is bonded from above by
\begin{eqnarray}
C_s^{(sum,up)} &\leq & \max \limits_{U \rightarrow X \rightarrow Y Z}  [I(U;Y) - \\ \nonumber
&& I(U;Z)].
\end{eqnarray}
\end{lem}

\begin{proof} Similar to the scenario discussed in Lemma 4, the CSO as well as the eavesdroppers  benefit by aggregating the information leaked from all uplink channels which leads to the following result:

\begin{eqnarray}
&&C_s^{(sum,up)} = \sum_{i=1}^{N}C_s^{(i,up)}  \\ \nonumber
&=& \sum_{i=1}^{N} \max  \limits_{U_i \rightarrow X_i \rightarrow Y_i Z_i} [I(U_{i};Y_{i})-I(U_i; Z_{[1,\ldots,N]})] \\ \nonumber
& \leq & \max \limits_{U_i \rightarrow X_i \rightarrow Y_i Z_i} [ \sum_{i=1}^{N}I(U_{i};Y_{i})- \sum_{i=1}^{N}I(U_{i}; Z_{[1,\ldots,N]}) ] \\ \nonumber
& \leq & \max \limits_{U \rightarrow X \rightarrow Y Z} [I(U;Y)-I(U; Z)]
\end{eqnarray}

\end{proof}
\subsection{Duality between CSO uplink and downlink channels}

\begin{lem}
The sum secrecy capacities of CSO downlink BC channel and uplink MAC are equivalent.
\end{lem}

\begin{proof}
Lemma 4 and Lemma 6 demonstrate the equivalence between the sum secrecy capacities of the  CSO uplink and downlink channels.
\end{proof}

We conjecture that if agents are allowed to cooperate among themselves, there exists a perfect duality between the secrecy capacities of CSO uplink and downlink channels. Detailed study of this duality is outside the scope of this paper and will be considered in future.

\subsection{CSO Problem with Parallel Gaussian channels}
Assume that all channels are additive white Gaussian noise (AWGN) channels as depicted in Fig.
\ref{fig2}.

In this case, the AWGN channel between the CSO and any one of the
agent $i$ can be described as
\begin{equation}
Y_{i} = X_{i} + N_{Mi}
\end{equation}
where $X_{i}$ and $Y_i$ are the input and the output of the
$i^{th}$ channel and $N_{Mi}$ represents the zero-mean white Gaussian noise, $i.e.$ $N_{Mi} ~\scriptsize{\sim}~ \mathcal{N}(0,\sigma_{N_{Mi}}^2)$, on the same
channel. Similarly, the eavesdropper's channel can be modelled as
\begin{equation}
Z_{i} = X_{i} + N_{Wi}
\end{equation}
where $Z_i$ represents the output of the eavesdropper's channel and
$N_{Wi}$ denotes zero-mean white Gaussian noise, $i.e.$ $N_{Wi} ~\scriptsize{\sim}~ \mathcal{N}(0,\sigma_{N_{Wi}}^2)$, on the eavesdropper's channel. We  assume that the model has a constraint on average power ($P$)
which an be described by
\begin{equation}
\frac{1}{2} \sum_{i=1}^{N}x_{i}^{2} \leq P.
\end{equation}

\begin{lem}
In the CSO model with parallel Gaussian channels, the
sum secrecy capacity $C_s^{(sum,down,G)}$ of the downlink BC channel between
the agents and the CSO has an upper bound given by:

\begin{equation}
C_s^{(sum,down,G)} \leq \sum_{i}^{N}\frac{1}{2}\log(1 +
\frac{P_i}{\sigma_{N_{Mi}}^2}) - \sum_{i}^{N}\frac{1}{2}\log(1 +
\frac{P_i}{\sigma_{N_{Wi}}^2}).
\end{equation}
\end{lem}

\begin{proof}
This result is obtained by substituting capacity of parallel
Gaussian channels in (\ref{Cssum}) \cite{liang1}.
\end{proof}

Next, we look into the power allocation strategies for the CSO model with parallel AWGN downlink channels. Power allocation for parallel secrecy channels
with Gaussian noise has been investigated in
\cite{zhangli}. The following theorem relates this to the CSO downlink channel and provides further insights.

\begin{theo}
The power allocation strategy that achieves
the sum secrecy capacity $C_s^{(sum,down,G)}$  is a function of both
the sum of variances $\sigma_{N_{Mi}}^2+\sigma_{N_{Wi}}^2$ and difference of variances $\sigma_{N_{Mi}}^2-\sigma_{N_{Wi}}^2$ and it
satisfies the following two constraints \cite{zhangli}:

\begin{eqnarray}
\label{eq:lambda}
P_i &>& 0, ~if~ \sigma_{N_{Wi}}^2 > \sigma_{N_{Mi}}^2 \\ \nonumber  &=& 0, ~otherwise
\\ \nonumber
\left( \frac{1}{\sigma_{N_{Mi}}^2} - \frac{1}{\sigma_{N_{Wi}}^2} \right) &>& 2 \lambda
\end{eqnarray}
where $\lambda$ is chosen such that $\sum_i \left[ \lambda -
(\sigma_{N_{Wi}}^2 - \sigma_{N_{Mi}}^2) \right]^+ = P$.
\end{theo}

\begin{proof}

We begin with the expression for the sum capacity of the channels
between the agents and the CSO. Maximization of the sum capacity
subject to the constraint on the total power can be solved using
Lagrange multiplier method. Writing the Lagrangian functional (J) as

\begin{eqnarray}
&& J(P_1,\ldots, P_N) = \sum_{i=1}^{N} \frac{1}{2}\log\left( 1+
\frac{P_i}{\sigma_{N_{Wi}}^2}\right) \\ \nonumber
&-& \sum_{i=1}^{N} \frac{1}{2}\log\left( 1+
\frac{P_i}{\sigma_{N_{Mi}}^2}\right) + \lambda(P - \sum_{i=1}^{N} P_i)
\end{eqnarray}

Differentiating this functional w.r.t. $P_i$, setting the
resultant expression equal to zero, and manipulating the equation,
we get
\begin{equation}
(P_i + \sigma_{N_{Mi}}^2)(P_i + \sigma_{N_{Wi}}^2) = \left(\frac{\sigma_{N_{Wi}}^2 -
\sigma_{N_{Mi}}^2}{2\lambda} \right).
\end{equation}
Applying Kuhn-Tucker conditions \cite{cover2} to make sure that $P_i > 0$, gives
us the first condition for power allocation strategy, which is
$N_{Wi} > N_{Mi}$. Continuing further and finding the roots of the
quadratic equation, results in the following condition:
\begin{equation}
P_i = \frac{1}{2} \left(\sqrt{N_{\triangle}^2 +
2\frac{N_{\triangle}}{\lambda}} - N_{\sum} \right)
\end{equation}
where $N_{\triangle}=\sigma_{N_{Wi}}^2 - \sigma_{N_{Mi}}^2$ and $N_{\sum}=\sigma_{N_{Wi}}^2 +
\sigma_{N_{Mi}}^2$. Applying Kuhn-Tucker condition again, we get the second
condition for power allocation strategy, which is,
\begin{equation}
\left( \frac{1}{\sigma_{N_{Mi}}^2} - \frac{1}{\sigma_{N_{Wi}}^2} \right) > 2 \lambda.
\end{equation}
\end{proof}

This theorem suggests that only those channels for which $(\sigma_{N_{Wi}}^2 - \sigma_{N_{Mi}}^2) > 0 $ need to be allocated power. If the threshold $\lambda$ is selected according to (\ref{eq:lambda}), then the power allocation strategy closely follows the classical waterfilling strategy. The difference is that in the CSO problem with parallel AWGN  channels, the channels are ranked based on $(\sigma_{N_{Wi}}^2 - \sigma_{N_{Mi}}^2)$. The CSO needs to communicate with its agents only on those channels which have an advantage over eavesdroppers and use the strategy described in Theorem 1 to invest on its resources.

\subsection{CSO Problem with Fading Downlink Channels}

A fading channel  between the CSO and any agent $i$ can be described \cite{zhangli} as
\begin{equation}
Y_{i}(j) = g_{Mi}(j)X_{i}(j) + N_{Mi}(j)
\end{equation}
where $X_{i}(j)$ and $Y_i(j)$ are the input and the output of the
$i^{th}$ channel during the time interval $j$ and $N_{Mi}(j)$ represents the noise on the same
channel. This is shown in Fig. \ref{fig3}. Similarly, the eavesdropper's channel can be modeled as
\begin{equation}
Z_{i}(j) = g_{Wi}(j)X_{i}(j) + N_{Wi}(j)
\end{equation}
where $Z_i(j)$ represents the output of the eavesdropper's channel during the interval $j$, and
$N_{Wi}(j)$ denotes noise on the eavesdropper's channel. The channel fading on the $i^{th}$ main channel and the $i^{th}$ eavesdropper's channel are represented by $g_{Mi}(j)$ and $g_{Wi}(j)$ which are zero-mean white Gaussian random variables, $i.e.$ $g_{Mi}(j) ~\scriptsize{\sim}~ N(0,\sigma_{g_{Mi}(j)}^2)$, $g_{Wi}(j) ~\scriptsize{\sim}~ N(0,\sigma_{g_{Wi}(j)}^2)$. Let us assume that $N_{Mi}(j)$ as well as $N_{Wi}(j)$ represent i.i.d. additive Gaussian noise with unit variance. For simplicity of description, we drop the index $(j)$ for time interval from here onwards.

We assume that each channel uses {\it{unit power}} for transmission as in \cite{liang1}, so we can denote the fading power gains corresponding to the two channels by $a_i =\sigma_{g_{Mi}}^2$ and $b_i=\sigma_{g_{Wi}}^2$ respectively and let $\gamma_i=(a_i,b_i)$ denote the channel state information (CSI). Let $P(\gamma_i)$ represent the power allocation for $i^{th}$ channel. The average power ($\hat{P}$) constraint for channel $i$ is given by $E[P(\gamma_i)] \leq \hat{P}$.

\begin{lem}
In a CSO model with fading downlink channels, if the channel state information is known, the
secrecy capacity $C_s^{(i,down,F)}$ of the $i^{th}$ downlink channel between the CSO and agent $i$ is given by:
\label{fading-lem2}

\begin{equation}
C_s^{(i,down,F)} = \underset{P(\gamma_i):E[P(\gamma_i)] =\hat{P}}{\max} E_{\gamma_i}[C(\gamma_i,P(\gamma_i)]
\end{equation}
where $C(\gamma_i,P(\gamma_i))$ = $\frac{1}{2}[\log(1+P(\gamma_i)a_i) - \log(1+P(\gamma_i)b_i)]$.
\end{lem}

\begin{proof} The proof for this lemma is based on the secrecy capacity of independent fading channels that can be found in \cite{gopala1,zhangli}. \end{proof}

\begin{theo}: In a CSO model with fading downlink channels, if the channel state information is known to all parties, the power allocation strategy that achieves
the secrecy capacity described in Lemma \ref{fading-lem2}
satisfies the following constraints \cite{zhangli}.
\begin{eqnarray}
\label{eq:fading}
P_i(a_i,b_i) &>& 0 ~if~ a_i > b_i \\ \nonumber  &=& 0 ~otherwise
\\ \nonumber
P(a_i,b_i) &=& \frac{1}{2} \left(\sqrt{N_{\triangle}^2 +
2\frac{N_{\triangle}}{\lambda}} - N_{\sum} \right)
\end{eqnarray}
where $N_{\triangle}=(\frac{1}{a_i} - \frac{1}{b_i})$ and $N_{\sum}= (\frac{1}{a_i} + \frac{1}{b_i})$.
The threshold $\lambda$ is chosen such that the average power constraint $E[P(\gamma_i)] \leq \hat{P}$ is satisfied.
\end{theo}

\begin{proof}
The proof for Theorem 2 is similar to that of Theorem 1 and can be found in \cite{zhangli}. The interpretation of the result is also similar to that of Theorem 1. The difference between the two is in terms of the dynamic nature of the channels. In the later case, the CSO needs to assess the channels more frequently to find the channels on which it has an advantage over the eavesdroppers, and allocate its resources accordingly. 
\end{proof}

\section{CSO Problem with Cooperating Agents}
In this section, we extend the CSO framework to include cooperation among the agents and investigate the advantages of such cooperation. The intuitive idea is that the agents that are not capable  of secret communications with the CSO may help other agents to enhance the secrecy capacity.  For example, two agents can cooperate with each other to make one of their channels capable of secret sharing. This can happen through jamming the eavesdropper just like the way the defense team attacks the offense team in football or soccer game.


\subsection{Enhancing Secret Communication Through Cooperation}
\label{pairing}
Assume that all communication channels in the CSO model are fading channels as shown in Fig. \ref{fig3}.
Among the $N$ fading channels, we define the channels that satisfy the following condition:
\begin{equation}
\frac{a_i}{\sigma_{N_{Mi}}^2} > \frac{b_i}{\sigma_{N_{Wi}}^2} \label{Eq:condition}
\end{equation}
as qualified channels (capable of secret communication with the CSO), where $\sigma_{N_{Mi}}^2$ and $\sigma_{N_{Wi}}^2$ are the variances of $N_{Mi}$ and $N_{Wi}$. Here, $a_i$ and $b_i$ are equal to $\sigma_{g_{Mi}}^2$ and $\sigma_{g_{Wi}}^2$ respectively. The channels that do not satisfy condition (\ref{Eq:condition}) are considered disqualified.


In a scenario where agent cooperation is not allowed, no secret sharing will be possible on disqualified channels. However, when agents are allowed to cooperate, they can help one another in distracting the eavesdroppers through  \emph{jamming}. One approach to achieve such a cooperation is through appropriate pairing of agents. Suppose that there are $k$ number of disqualified channels and they are ordered according to their SNR values:

\begin{equation}
\frac{a_1}{\sigma_{N_{M1}}^2} < \frac{a_2}{\sigma_{N_{M2}}^2} < ...... < \frac{a_k}{\sigma_{N_{Mk}}^2}. \label{Eq:relation}
\end{equation}

There exist many strategies for pairing the cooperative nodes. One approach that  maximizes the number of {\it{secret communication channels}} is described below:
For each channel $i ~ \in ~(1, \ldots, k)$, find the smallest and unused channel $s ~ \in ~ (i+1, \ldots, k)$ such that $\frac{a_s}{\sigma_{N_{Ms}}^2} > \frac{b_{i}}{\sigma_{N_{Wi}}^2}$. If such $s$ exists, pair up agent $i$ with agent $s$. If such $s$ doesn't exist, move on to the next channel.

In the best case scenario, every agent will be successfully paired with another cooperating agent. This leads to qualifying  half of the originally disqualified channels.

The ratio $\frac{C_{sq,i}}{C_{q,i}}$ is used as a metric to evaluate the \emph{secrecy efficiency} of a qualified channel $i$, where $C_{sq,i} = \log_2(1+ \frac{a_i}{\sigma_{N_{Mi}}^2}) - \log_2(1+\frac{b_{i}}{\sigma_{N_{Wi}}^2})$ and $C_{q,i} = \log_2(1+ \frac{a_i}{\sigma_{N_{Mi}}^2})$. The denominator $C_{q,i}$ reflects the capacity of the channel $i$ in the absence of any eavesdropper.

When a disqualified channel $i$ is turned into a qualified channel with the help of a cooperating agent $h$, the eavesdropper listening to channel $i$ is completely jammed by  the cooperating agent $h$ provided $\frac{a_h}{\sigma_{N_{Mh}}^2} \geq \frac{b_i}{\sigma_{N_{Wi}}^2} \geq \frac{a_i}{\sigma_{N_{Mi}}^2}$. The \emph{secrecy efficiency} achieved through the pair of cooperating agents $(i,h)$ is given by $\frac{C_{sp,i}}{C_{p,i}}$, where $C_{sp,i} = \log_2(1+ \frac{a_i}{\sigma_{N_{Mi}}^2})$ and $C_{p,i} = \log_2(1+ \frac{a_i}{\sigma_{N_{Mi}}^2}) + \log_2(1+ \frac{a_h}{\sigma_{N_{Mh}}^2})$. The denominator $C_{p,i}$ denotes the sum of capacities of the channels corresponding to the two cooperating agents. Increasing the number of secret communication channels is equivalent to increasing the overall secrecy capacity of the CSO uplink channels.

\indent  Fig. \ref{fig4} illustrates the secrecy efficiency achieved through the proposed cooperation strategy. It shows a scenario in which there are three qualified channels and six disqualified channels. Half of the six disqualified channels are  turned into qualified channels through cooperation. Note that the maximum secrecy efficiency achieved through cooperation is 0.5.

\section{Summary}
This paper introduced the CSO problem and investigated two scenarios: one without cooperation  and the other with cooperation among agents. Secrecy capacity and power allocation strategies are discussed in both contexts. A strategy for enhancing secrecy capacity through cooperation is proposed and its optimality in terms of maximizing the secrecy capacity has been proved. Experimental results are provided to illustrate the benefits of the proposed strategy. Other  models of cooperation that limit the information leaked to eavesdropper's will be explored in future. 

\begin{appendices}
\section{Optimality of the Proposed Pairing Strategy}
In this appendix, we show that the proposed pairing strategy described in Section \ref{pairing} is optimal in terms of maximizing the number of secret communication channels.

\begin{proof}
Let $A_{i}$, $i ~ \in ~(1, \ldots, k)$ represent the SNR of the disqualified channel from agent $i$  to CSO, i.e. $A_{i} = \frac{a_i}{\sigma_{N_{Mi}}^2}$. Let $E_{i}$, $i ~ \in ~(1, \ldots, k)$ represent the SNR of the corresponding eavesdropper channel, i.e. $E_{i} = \frac{b_{i}}{\sigma_{N_{Wi}}^2}$. A channel $i$ becomes disqualified when $A_{i} < E_{i}$. Then, the node pairing problem is equivalent to choosing a cooperating node $j$ such that ($A_{j} > E_{i} > A_{i}$). Then, the node pairing problem is equivalent to finding as many ($i$, $j$) pairs as possible that satisfy this condition.

Let $S(i)$ denote the set of all nodes that can pair up with node $i$ and let $|S(i)|$ represent the size of this set.  Assume that $A_{i}$'s are sorted from the smallest to the largest. Since $A_{k}$ is the largest among all $A_{i}$s, $|S(k)|$ = 0. An example consisting of five agent-eavesdropper pairs is shown in Fig. \ref{fig5}. In this example, since $A_{5} > A_{4} > A_{3} > A_{2} > E_{1} > A_{1}$, the set for agent $1$ is $S(1) = \{2, 3, 4, 5\}$, and $|S(1)| = 4$. Also, it is easy to see that $S(2) = \{4, 5\}$, $S(3) = \{4, 5\}$, $S(4) = \{5\}$, $S(5) = \{\phi \}$, and $|S(2)| = 2$, $|S(3)| = 2$, $|S(4)| = 1$, $|S(5)| = 0$.

A reasonable assumption is that each agent is aware of how strong (channel SNR) his/her corresponding eavesdropper is. Further, since agents cooperate with each other, it is reasonable to assume that all agents share their eavesdropper's channel state information (CSI) with one another. Then, three different cases arise as discussed below.

\noindent {\it{Case 1}}: \begin{displaymath}
E_{1} \leq E_{2} \leq ...... \leq E_{k}
\end{displaymath}
In this case, all agents who can pair up with  $j$  can also pair up with $i$, for all $j > i$, since $E_{j} > E_{i}$. Therefore,
\begin{displaymath}
|S(1)| \geq |S(2)| \geq ...... \geq |S(k-1)|.
\end{displaymath}

If there exists an $x$, $1 \leq x \leq k-1$,  such that  $|S(x)| = 1$, and $|S(x-1)| > 1$. This implies that $x$ can pair up only with one agent which is $k$. Starting from agent $1$, if all agents before $x$ randomly pick one cooperating agent from their respective sets, then the probability ($Pr(picking ~ k)$) of  $k$ being picked by any one of them resulting in  $x$ being not able to pair up with any agent is given by:
\begin{displaymath}
\begin{array}{lllll}
Pr(picking ~k) = 1 - Pr(not ~picking ~k)\\
= 1 - \frac{C_{|S(1)|-1}^{1}}{C_{|S(1)|}^{1}} \cdot \frac{C_{|S(2)|-1}^{1}}{C_{|S(2)|}^{1}} \cdot ... \cdot \frac{C_{|S(x-1)|-1}^{1}}{C_{|S(x-1)|}^{1}}\\
\\
= 1 - \frac{|S(1)|-1}{|S(1)|} \cdot \frac{|S(2)|-1}{|S(2)|} \cdot ... \cdot \frac{|S(x-1)|-1}{|S(x-1)|}.
\end{array}
\end{displaymath}
This guarantees that $Pr(picking ~k) > 0$. On the other hand, if the proposed pairing strategy is used, then, all agents before $x$ pick the {\it{smallest}} $i$ to pair up with from their respective sets, then, $Pr(picking \quad k) = 0$. Thus, the proposed pairing strategy is optimal for {\it{case 1}}, because it guarantees that agent $x$ can pair up with someone even though he/she has the least number of options.

\noindent {\it{Case 2}}: \begin{displaymath}
E_{1} > E_{2} > ...... > E_{k}
\end{displaymath}
In this case, all agents who can pair up with $i$  can also pair up with $j$, for all $j > i$, since $E_{j} < E_{i}$. Therefore,
\begin{displaymath}
|S(1)| \leq |S(2)| \leq ...... \leq |S(k-1)|.
\end{displaymath}
In this case, agent $i$ is not affected by the choices made by others before $i$, $i ~ \in ~(1, \ldots, k)$. Hence, no pairing strategy is needed. However, if the proposed pairing strategy used, it achieves the highest \emph{secrecy efficiency}. For example, assume the set for agent $m$ is $S(m) = \{l, l+1, ... ,k \}$, where $A_{m} < A_{l} < A_{l+1} < ... < A_{k}$. If a node $d$ is randomly selected from $S(m)$ to pair up with $m$, then $A_{d} \geq A_{l}$ and the \emph{secrecy efficiency} is $\frac{A_{m}}{A_{m}+A_{d}}$. If the smallest agent is selected from $S(m)$ using the proposed strategy, then the \emph{secrecy efficiency} is $\frac{A_{m}}{A_{m}+A_{l}}$. Since $A_{d} \geq A_{l}$, $\frac{A_{m}}{A_{m}+A_{d}} \leq \frac{A_{m}}{A_{m}+A_{l}}$. Therefore, the proposed pairing strategy is optimal for {\it{case 2}}, because it achieves the highest \emph{secrecy efficiency}.

\noindent {\it{Case 3}}: There is no particular order of $E_{i}s$.

This is a random combination of {\it{case 1}} and {\it{case 2}}. Therefore, the proposed pairing strategy is an optimal in this case also.
\end{proof}
\end{appendices}

\singlespacing

\bibliographystyle{IEEETrans}
\bibliography{IEEEabrv,cso_bib}

%
%
%

\newpage

\begin{figure}
\begin{center}
\includegraphics[width=3in]{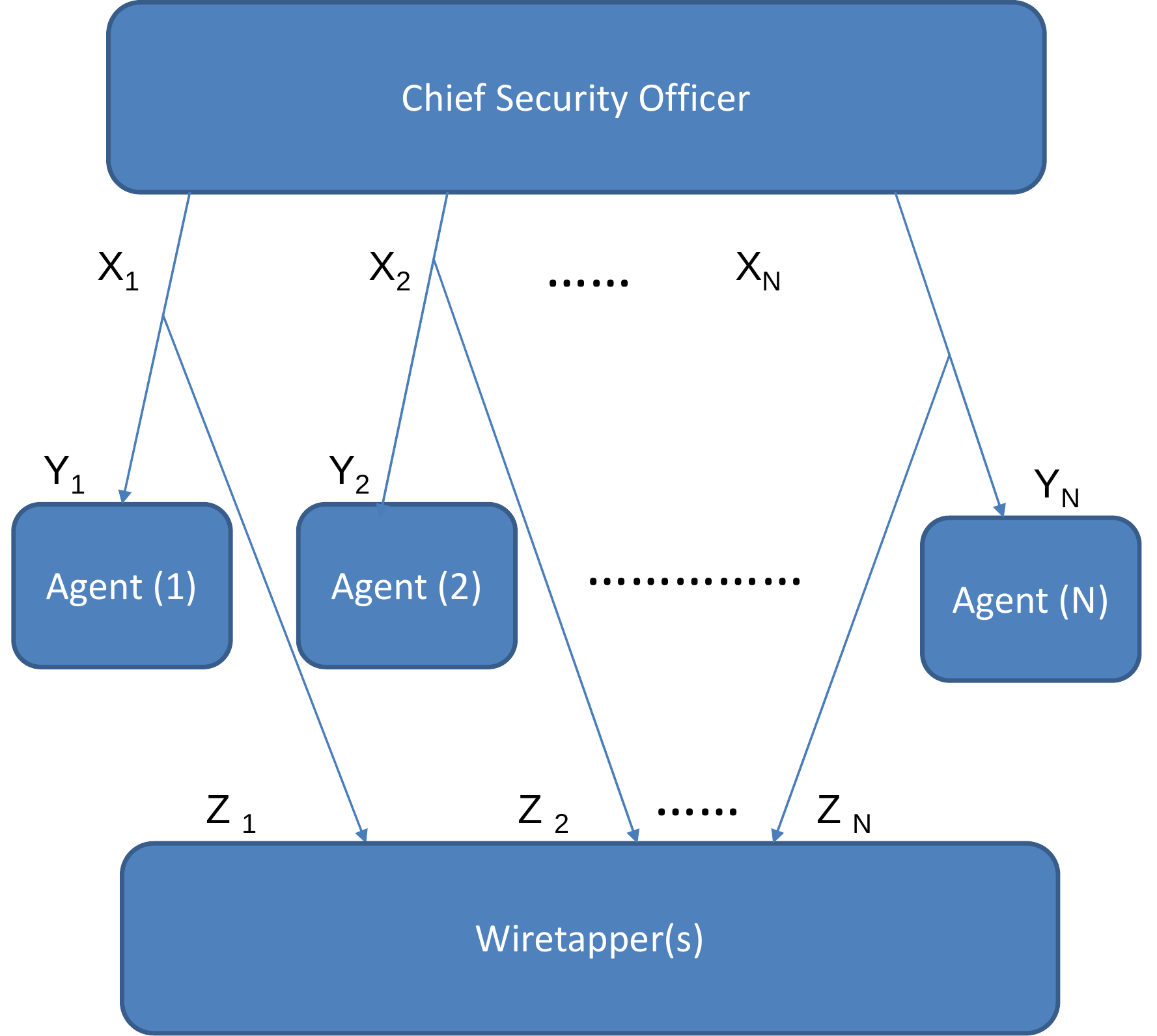}
\caption{The CSO framework models the secret information exchange between the CSO and the  agents.} \label{fig1}
\end{center}
\end{figure}

\begin{figure}
\begin{center}
\includegraphics[width=3in]{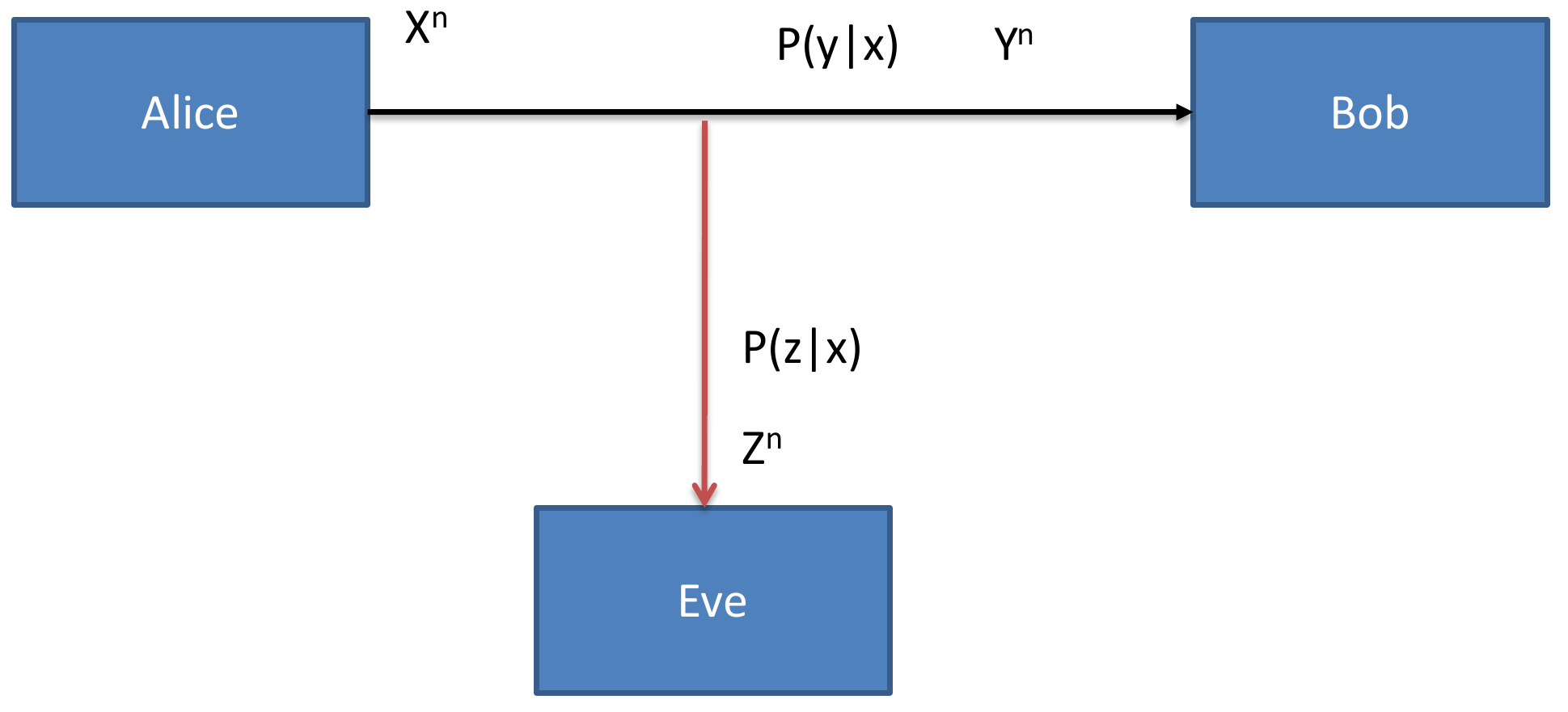}
\caption{Alice-Bob-Eve secrecy channel} \label{fone}
\end{center}
\end{figure}

\begin{figure}
\begin{center}
\includegraphics[width=3in]{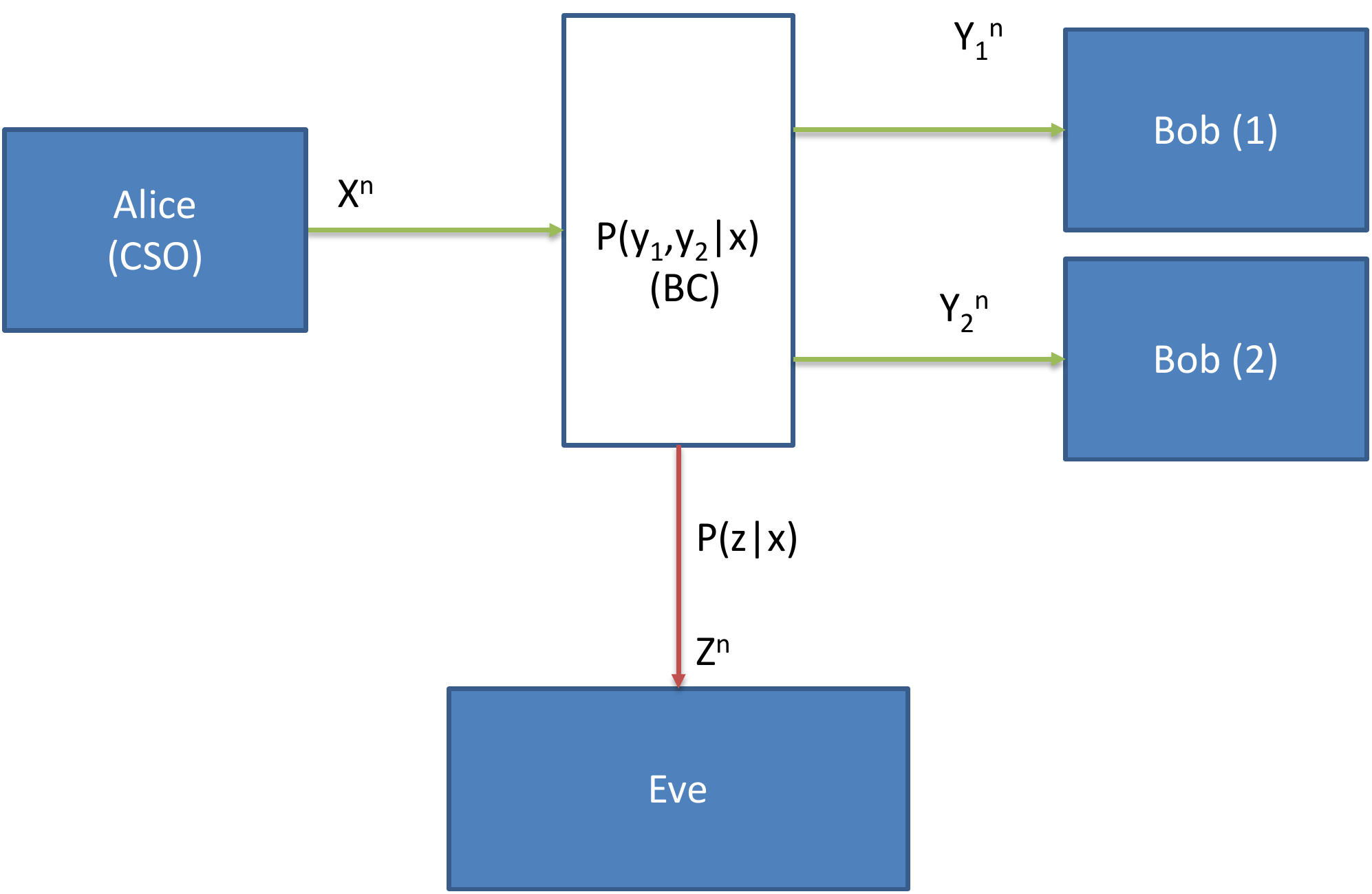}
\caption{CSO downlink broadcast channel (BC) with two agents} \label{ftwo}
\end{center}
\end{figure}

\begin{figure}
\begin{center}
\includegraphics[width=3in]{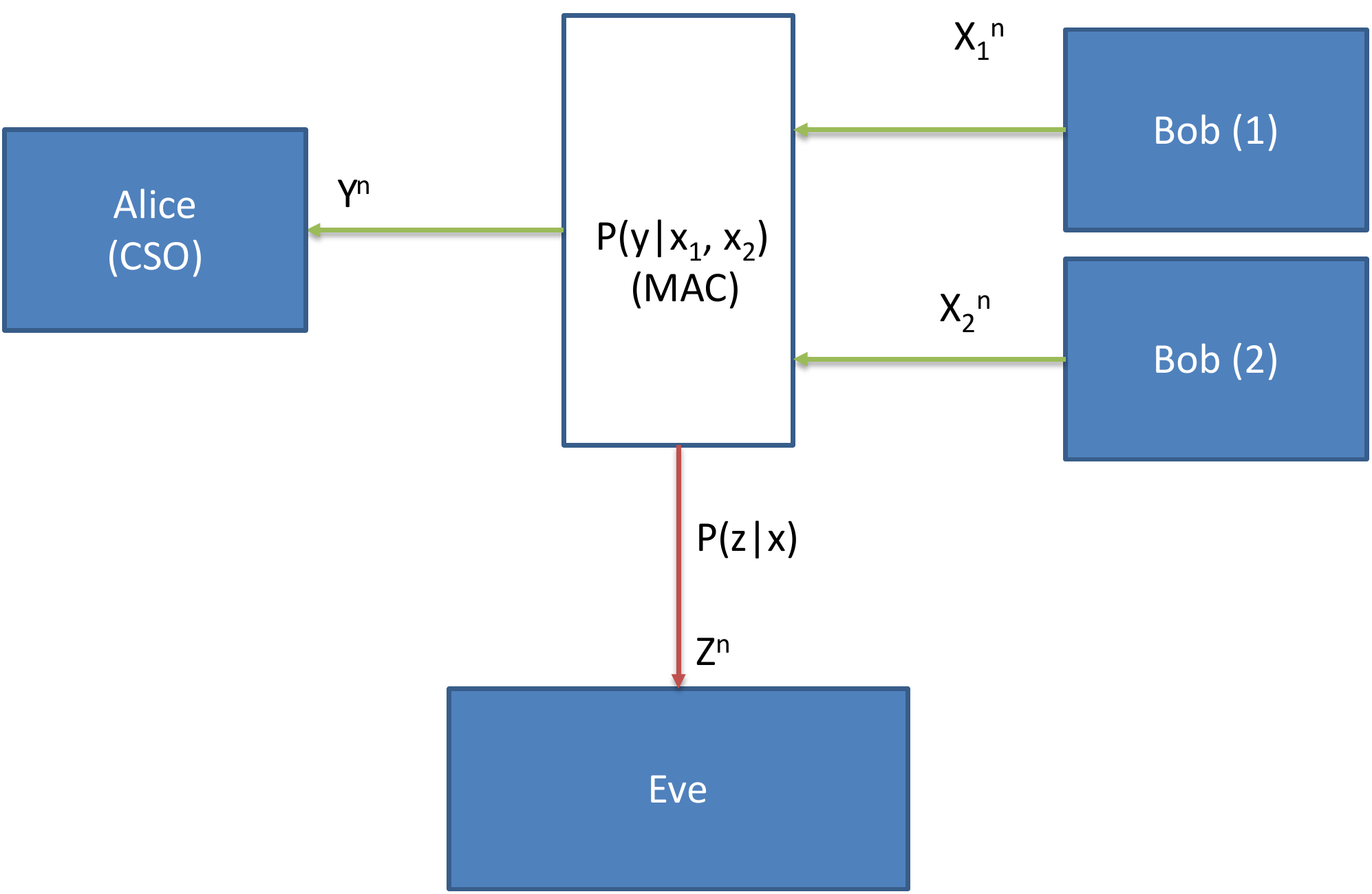}
\caption{CSO uplink multiple access channel (MAC) with two agents} \label{fthree}
\end{center}
\end{figure}

\begin{figure}[htbp]
\begin{center}
\includegraphics[width=3in]{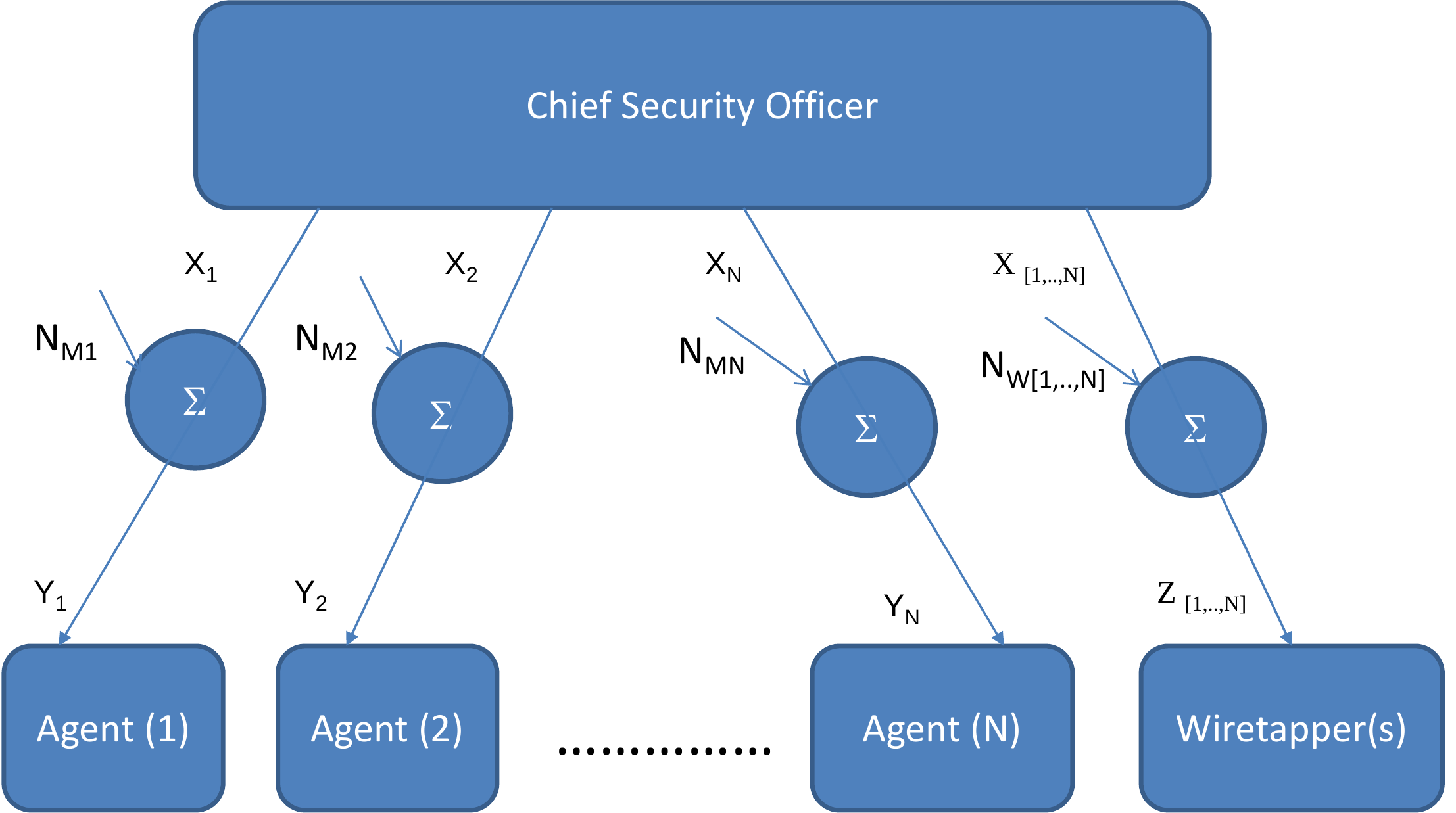}
\caption{The CSO problem with AWGN downlink channels.} \label{fig2}
\end{center}
\end{figure}

\begin{figure}[htbp]
\begin{center}
\includegraphics[width=3.5in]{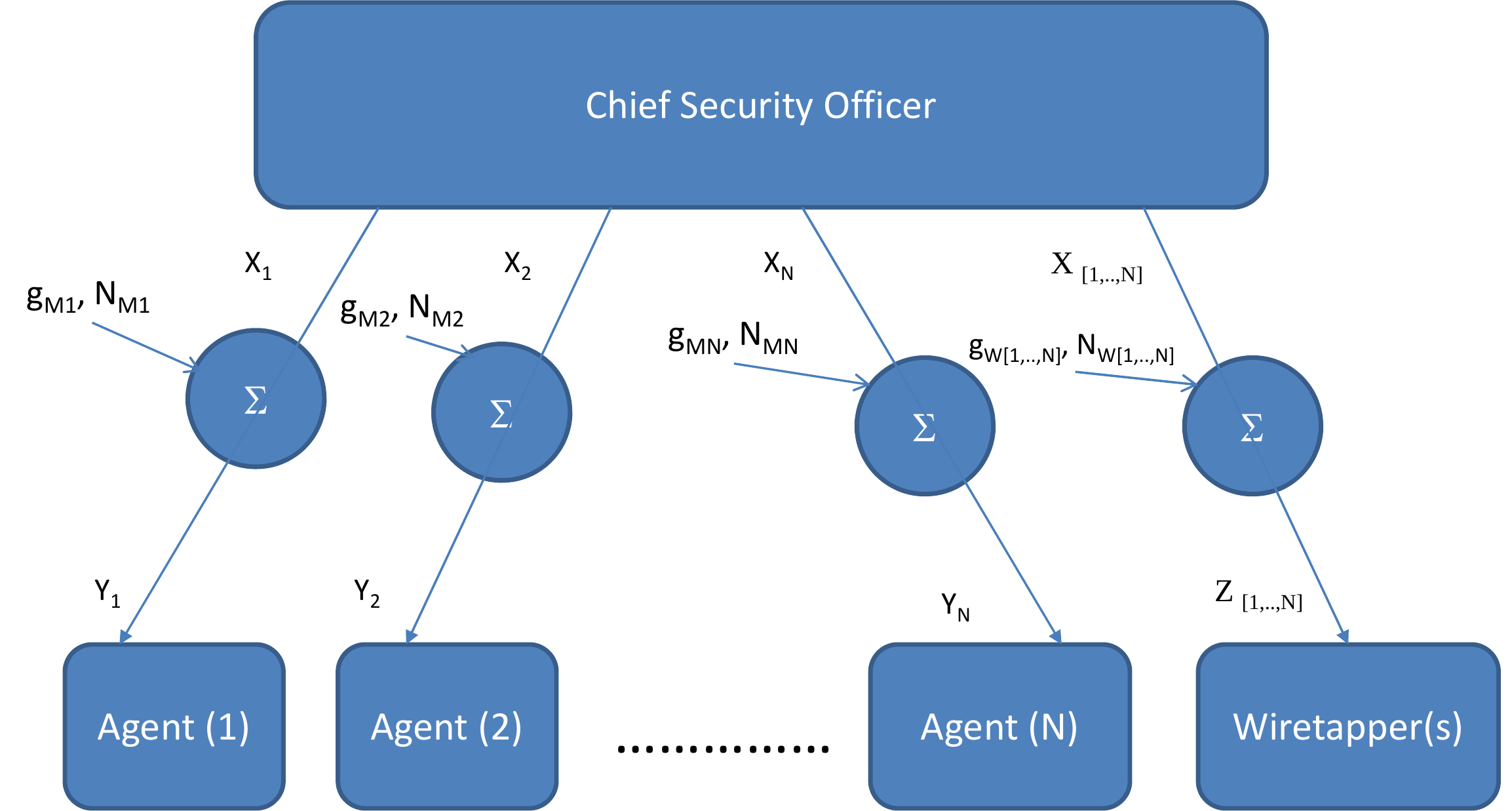}
\caption{The CSO problem with fading
downlink channels.} \label{fig3}
\end{center}
\end{figure}

\begin{figure}[htbp]
\begin{center}
\includegraphics[width=3.5in]{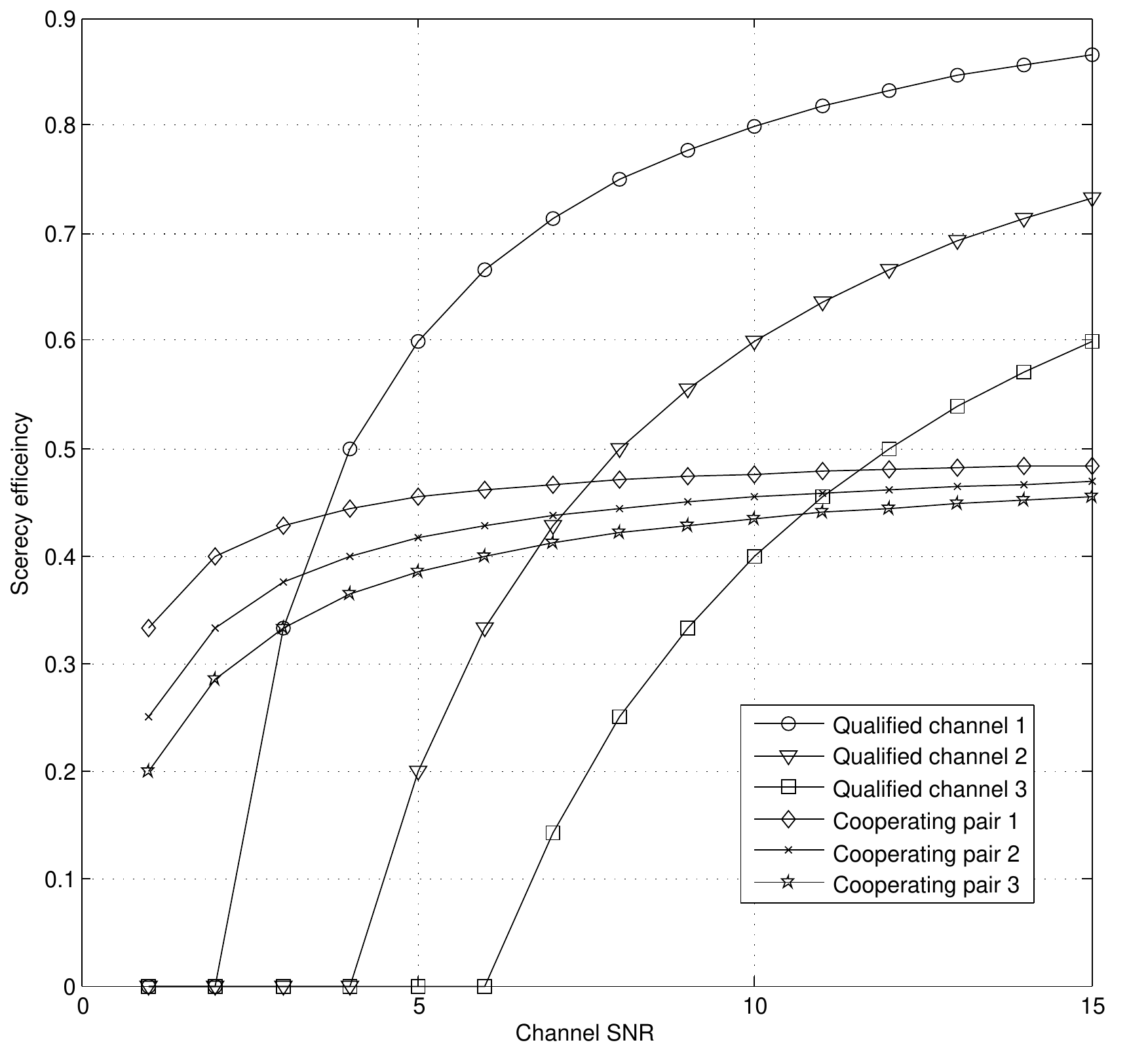}
\caption{Using cooperation among agents to increase the secrecy capacity of CSO fading downlink channel. The three plots corresponding to the qualified channels reflect the effect of eavesdroppers with low, medium, and strong  channel SNRs. The three plots corresponding to the cooperating pairs reflect the effect of cooperation from agents with low, medium, and strong  channel SNRs.} \label{fig4}
\end{center}
\end{figure}

\begin{figure}
\begin{center}
\includegraphics[width=3.5in]{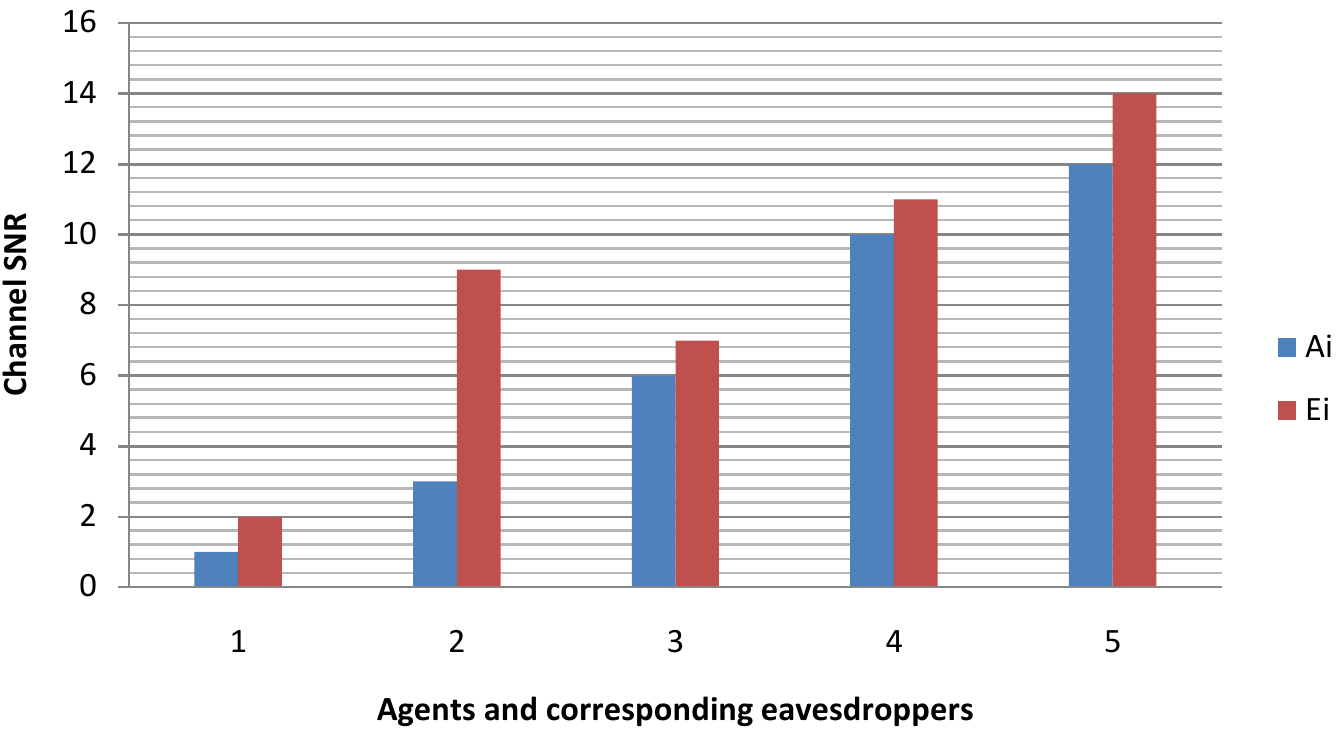}
\caption{An example of the relation between agent and its corresponding eavesdropper} \label{fig5}
\end{center}
\end{figure}

\end{document}